\documentclass[authoryear,preprint,10pt]{elsarticle}

\usepackage{amsmath}
\usepackage{amssymb}
\usepackage{graphicx}
\usepackage{cite}
\usepackage{color}
\usepackage{lineno}
\usepackage{natbib}
\newtheorem{thm}{Theorem}

\newtheorem{proof}{Proof}

\usepackage{setspace}
\usepackage{ulem}

\journal{Journal of Theoretical Biology}

\begin{document}

\begin{frontmatter}
\title{Cancer initiation with epistatic interactions between driver and passenger mutations}

\author[Ploen]{Benedikt Bauer}
\author[Kiel]{Reiner Siebert}
\author[Ploen]{Arne Traulsen}

\address[Ploen]{Evolutionary Theory Group, Max Planck Institute for Evolutionary Biology, August-Thienemannstra\ss e 2, 24306 Pl\"{o}n, Germany}
\address[Kiel]{Institute of Human Genetics, Christian-Albrechts-University Kiel \& University Hospital Schleswig-Holstein, Campus Kiel, Schwanenweg 24, 24105 Kiel, Germany}

\begin{abstract}
We investigate the dynamics of cancer initiation in a mathematical model with one driver mutation and several passenger mutations. 
Our analysis is based on a multi type branching process: We model individual cells which can either divide or undergo apoptosis. 
In case of a cell division, the two daughter cells can mutate, which potentially confers a change in fitness to the cell. 
In contrast to previous models, the change in fitness induced by the driver mutation depends on the genetic context of the cell, 
in our case on the number of passenger mutations. 
The passenger mutations themselves have no or only a very small impact on the cell's fitness. 
While our model is not designed as a specific model for a particular cancer, the underlying idea is motivated by clinical and experimental observations in Burkitt Lymphoma. 
In this tumor, the hallmark mutation leads to deregulation of the {\it MYC} oncogene which increases the rate of apoptosis, but also the proliferation rate of cells.
This increase in the rate of apoptosis hence needs to be overcome by mutations affecting apoptotic pathways, naturally leading to an epistatic fitness landscape.
This model shows a very interesting dynamical behavior which is distinct from the dynamics of cancer initiation in the absence of epistasis. 
Since the driver mutation is deleterious to a cell with only a few passenger mutations, there is a period of stasis in the number of cells until a clone of cells with enough passenger mutations emerges. 
Only when the driver mutation occurs in one of those cells, the cell population starts to grow rapidly.
\end{abstract}

\begin{keyword}
Cancer \sep modeling \sep somatic evolution \sep population dynamics 

\end{keyword}

\end{frontmatter}


\journal{Journal of Theoretical Biology}

\section{Introduction}

Tumors develop by accumulating different mutations within a cell, which affect the cell's reproductive fitness 
\citep{armitage:BJC:1954,lengauer:Nature:1998,hanahan:Cell:2000,michor:NRC:2004,wodarz:book:2005,sjoblom:Science:2006,greenman:Nature:2007,wood:Science:2007,Jones:PNAS:2008,attolini:ANAS:2009,parmigiani:Genomics:2009,gerstung:MPS:2010}. 
As \citet{bozic:PNAS:2010}, we refer to the fitness of a mutated cell as the ratio between a cell's rate to proliferate and the cell's rate of apoptosis compared to wild type cells. 
The higher the fitness, the more likely it is for the cell to proliferate. 
For high fitness values, the population of cells growths very fast and stochastic effects play a minor role. 
In our model, this can be thought of as the formation of a tumor.

However, many mutations have no impact on the cell's fitness, e.g.\ mutations not affecting coding or regulatory sequences. 
Other mutations may lead to a fitness disadvantage, which implies that the cell's risk of apoptosis is higher than its chance of proliferation. However, the same mutations in combination with other mutations within the same cell might lead to a large fitness advantage.

We were motivated by genetic studies in Burkitt Lymphoma, 
a highly aggressive tumor,
where a single genetic alteration has an impact on a wide range of other genes, some of them affect cell growth while others induce apoptosis. 
More specifically, a chromosomal translocation between the {\it MYC} protooncogene on chromosome 8 and one of three immunoglobulin ({\it IG}) genes is found in almost every case of Burkitt Lymphoma 
\citep{richter:NG:2012,allday:SCB:2009,hummel:NEJM:2006,sander:CC:2012}. 
This leads to deregulated expression of the {\it MYC} RNA and in consequence, to deregulated MYC protein expression. The MYC protein acts as a transcription factor and has recently been shown to be a general amplifier of gene expression \citep{nie:Cell:2012,lin:Cell:2012}, targeting a wide range of different genes. 
Most importantly, {\it MYC} expression induces cell proliferation. 
In Burkitt Lymphoma, the {\it IG-MYC} fusion
is evidently the key mutation for tumorigenesis
\citep{salaverria:JCO:2011,zech:IJC:1976,schmitz:cshp:2014,campo:NG:2012}.
However, {\it MYC} plays also a key role in inducing apoptosis \citep{pelengaris:Cell:2002,wang:CBT:2011,hoffman:Oncogene:2008}.
Thus, the {\it IG-MYC} translocation alone would lead to cell death. 
Therefore, the {\it IG-MYC} translocation has to be accompanied by additional mutations, which deregulate the apoptosis pathways, such as mutations affecting e.g.\ {\it TP53} or {\it ARF} \citep{richter:NG:2012,allday:SCB:2009,sander:CC:2012}. 
These additional mutations have probably only little direct impact on the cell's fitness, since apoptosis is rare.
Hence, these mutations cannot be considered as primary driver mutations in the context of Burkitt Lymphoma. 
However, in combination with the {\it MYC} mutation these additional mutations decrease the apoptosis rate.
Consequently, the cells proliferate fast and the population grows accordingly, leading to tumorigenesis.
Because all cells carry the {\it MYC} mutation in Burkitt Lymphoma, but fast growth does not start immediately with that mutation,
it seems to confer its large fitness advantage only in a certain genetical context. 
Thus, interactions between different mutations may crucially affect the dynamics of cancer progression. 
Due to the fact that those additional mutations do not confer a direct fitness advantage, they cannot be considered as driver mutations.
Nevertheless, at least some of them are necessary in order for the {\it MYC} mutation to become advantageous for the cell.
Therefore, they cannot be regarded as true passenger mutations, either.
Throughout this manuscript, we therefore call these additional mutations ``secondary driver mutations''.

Besides Burkitt Lymphoma, epistatic effects in cancer initiation seem also to be relevant for other cancers. 
For example, we can think of the inactivation of a tumor suppressor gene discussed by Knudson in the context of retinoblastoma \citep{knudson:PNAS:1971}.
This inactivation is neutral for the first hit but highly advantageous for the second hit, 
and can hence be viewed as an interaction of genes \citep{nowak:PNAS:2002,nowak:PNAS:2004,vogelstein:NM:2004,iwasa:JTB:2005}. 
Another case is found in lung carcinomas, where activation of each of two oncogenes ({\it SOX2} and {\it PRKCI}) alone is insufficient,
but in concert they initiate cancer \citep{justilien:CC:2014}.
In other cases, there is clear evidence for sign epistasis: 
The {\it ras} family of proto-oncogenes is also discussed to underlie epistatic effects. 
Amplification of {\it ras} leads to senescence in the cell. 
Nevertheless, {\it ras} is a well known oncogenic driver gene. 
Hence, the {\it ras} mutation needs to be accompanied by other mutations \citep{elgendy:MC:2011,serrano:CELL:1997}.
Moreover, the difficulty to distinguish between drivers and passengers \citep{futreal:CC:2007,frohling:CC:2007} suggests that
for a full understanding of cancer initiation it is insufficient to think of these two types of mutations only.  

So far, most models have focused on the idea that passenger mutations have no effect or only a little effect, 
whereas each driver mutation increases the fitness of the cell \citep{michor:NRC:2004,beerenwinkel:PlosCB:2007,bozic:PNAS:2010,gerstung:MPS:2010,antal:JSM:2011,reiter:EA:2013, durrett:TPB:2010, datta:EA:2013}. 
Other models focus on the neutral accumulation of mutations \citep{durrett:AnAP:2009, luebeck:PNAS:2002}. 
Moreover, different mutations are typically treated as independent, which is a strong assumption that will
often not be fulfilled. 
In our model, mutations are interacting in an epistatic way \citep{wolf:book:2000}: 
the change in fitness induced by the driver mutation depends strongly on the genetic environment, 
i.e.\ in our case on the number of secondary driver mutations that are present in that cell. 
In addition we assume that the secondary driver mutations alone have almost no fitness advantage. 
Such a dependence between mutations can strongly affect the dynamics of cancer initiation. 
In evolutionary biology, epistatic systems are often analyzed regarding the structure or ruggedness of the landscape and the accessibility of different pathways \citep{weinreich:Evolution:2005a,franke:PLosCB:2011,szendro:PNAS:2013}.
The experimental literature also studies which factors can lead to epistasis \citep{visser:PRSB:2011,szappanos:NG:2011}.
Here, we are interested in the dynamics of such an epistatic model, which we illustrate by stochastic, individual based simulations. 
In addition, we derive analytical results for the average number of cells with different combinations of mutations and find a good agreement with the average dynamics in individual based computer simulations.
Furthermore, we discuss the computation of the waiting time until cancer initiation.
Our results show that the dynamics in such systems of epistatic interactions are distinct from previous models of cancer initiation \citep{michor:NRC:2004,beerenwinkel:PlosCB:2007,bozic:PNAS:2010,gerstung:MPS:2010,antal:JSM:2011,reiter:EA:2013}, which may have important consequences for the treatment of such cancers. 
While in previous models there is a steady increase in growth with every new mutation, in our model there is a period of stasis followed by a rapid tumor growth.

Of course, the biology of Burkitt Lymphoma is much more complex than modelled herein. 
To make the model more realistic one would have to distinguish between the different secondary driver mutations, since different genes contribute differently to the cells fitness, especially in a cell where the {\it IG-MYC} fusion is present.
Our model is not aimed to realistically describe such a situation in detail.
Instead, we focus on the extreme case of the so called \textit{all-or-nothing} epistasis 
\citep{Barrick:NRG:2013,meyer:Science:2012} to illustrate its effect on the dynamics of cancer initiation. 
As there is no theoretical analysis of epistatic effects in cancer initiation so far, 
a well understood minimalistic model seems to be necessary in order to illustrate the potential impact of epistasis on cancer progression.
Our minimalistic model clearly shows that epistasis can lead to a qualitatively different dynamics of cancer initiation.

\section{Mathematical model}\label{sec:methods}

We analyze cancer initiation in a homogenous population of initially $N$ cells with discrete generations.
In every generation, each of the $N$ cells can either die or divide. 
If a cell divides, its two daughter cells can mutate with mutation probabilities 
$\mu_\mathrm{D}$ for the driver mutation and $\mu_\mathrm{P}$ for secondary driver mutations (where the $P$ indicates that these would be called passenger mutations in closely related models). 
In principle, we could drop the assumption that these two mutation probabilities are independent on the cell of origin, but this would lead to inconvenient notation.
We neglect back mutations and multiple mutations within one time step, because their probabilities are typically very small.
Figure~\ref{fig:model} summarizes the possible mutational pathways of the model. 
The variables $x_{i,j}$ denote the number of cells with or without the primary driver mutation ($i=1$ or $i=0$ respectively), and $j$ secondary driver mutations.
\begin{figure}[h]
\begin{center}
\includegraphics[width=0.5\textwidth]{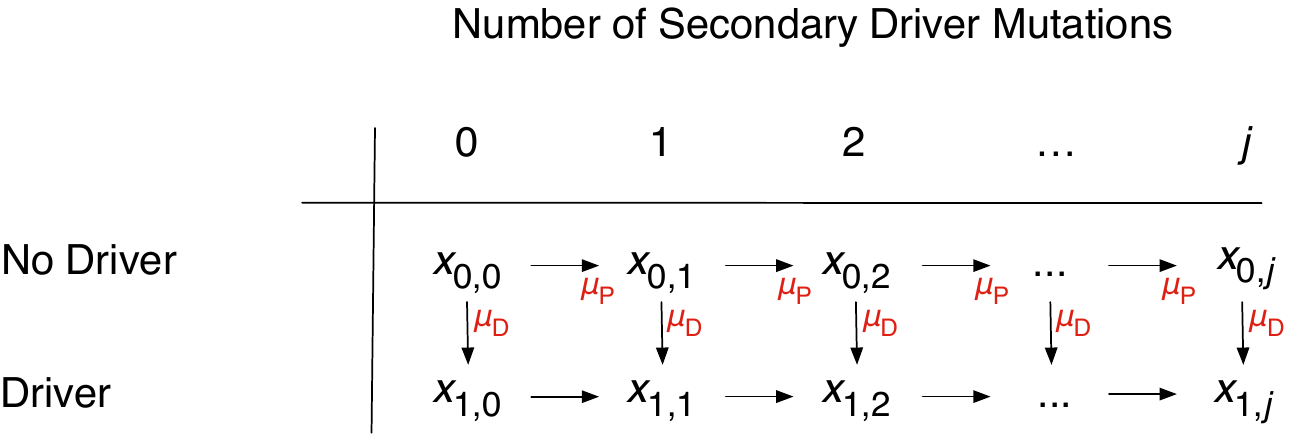}\\
\includegraphics[width=0.49\textwidth]{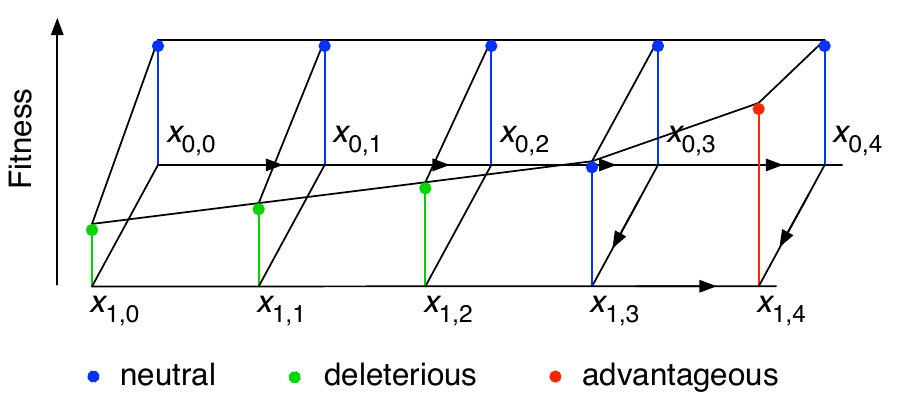}
\caption{Mutational pathways of the model. 
Left: The entries $x_{i,j}$ denote the number of cells with or without the primary driver mutation ($i=1$, or $i=0$ respectively), and $j$ secondary driver mutations. 
Right: Cells with only secondary driver mutations have neutral or nearly neutral fitness. 
The fitness of cells with the primary driver mutation depends on the number of secondary driver mutations within the cell, leading to an epistatic fitness landscape.}
\label{fig:model}
\end{center}
\end{figure}

A cell's probability for apoptosis and proliferation depends on the presence of the primary driver mutation and on the number of secondary driver mutations it has accumulated. 
For cells with no mutations, the division and apoptosis probabilities are both equal to $\tfrac{1}{2}$. 
This implies that the number of cells is constant on average as long as no further mutations occur. 
We assume that the initial number of cells is high and thus we can neglect that the population would go extinct \citep{haccou:book:2005}.
For our parameter values, the expected extinction time of our critical branching process
exceeds the average life time of the organism by far.

For cells without the primary driver mutation, each secondary driver mutation leads to a change in the cell's fitness by $s_\mathrm{P}$.
For cells with the primary driver mutation, the fitness advantage obtained with each secondary driver mutation is $s_\mathrm{DP}$. 
The driver mutation increases both the apoptosis rate and the proliferation rate.
The increase in the apoptosis rate is $s_{\mathrm{D}^-}$ and the increase in the division rate is $s_{\mathrm{D}^+}$. 
With these parameters, the proliferation rate 
$b_{0,j}$
for cells with $j$ secondary driver mutations but without the primary driver mutation is
\begin{align}
 b_{0,j} &= \frac{1}{2}(1+s_\mathrm{P})^j,
\label{eq:ratesb0}
\end{align}
whereas the proliferation rate 
$b_{1,j}$
for such cells with the primary driver mutation is
\begin{align}
 b_{1,j} &= \frac{1}{2}\cdot \frac{1+s_\mathrm{D}^+}{1+s_\mathrm{D}^-}(1+s_\mathrm{DP})^j.
 \label{eq:ratesb1}
\end{align}
The apoptosis rates, denoted as $d_{0,j}$ and $d_{1,j}$ are simply one minus the proliferation rate
\begin{align}
\begin{split}
d_{0,j} &= 1 - \frac{1}{2}(1+s_\mathrm{P})^j, \\
d_{1,j} &= 1- \frac{1}{2}\cdot \frac{1+s_\mathrm{D}^+}{1+s_\mathrm{D}^-} (1+s_\mathrm{DP})^j.
\end{split}
\label{eq:ratesd}
\end{align}

\section{Results}
\subsection{Simulations}

Mutations occur at fixed rates $\mu_\mathrm{D}$ and $\mu_\mathrm{P}$ for primary and secondary drivers, respectively. 
For a long time, the overall fitness  does not increase noticeably. 
For $s_\mathrm{P}=0$, it stays on average constant.
Hence, the total number of cells stays approximately constant. 
Only when a cell with enough secondary driver  mutations and also the primary driver mutation arises, the cell's fitness is increased substantially beyond the fitness of other cells and its chance of proliferation is significantly increased. 
At that point, the total number of cells starts to increase rapidly, see Figure \ref{fig:total}.
In models presented in literature so far, the cell's fitness is increased independently with every (driver) mutation \citep[see e.g.\ ][]{beerenwinkel:PlosCB:2007,bozic:PNAS:2010}.
Although the total number of cells increases exponentially, 
these models do not find a sudden burst in the number of cells. 
\begin{figure}[h]
\begin{center}
 \includegraphics[width=0.8\textwidth]{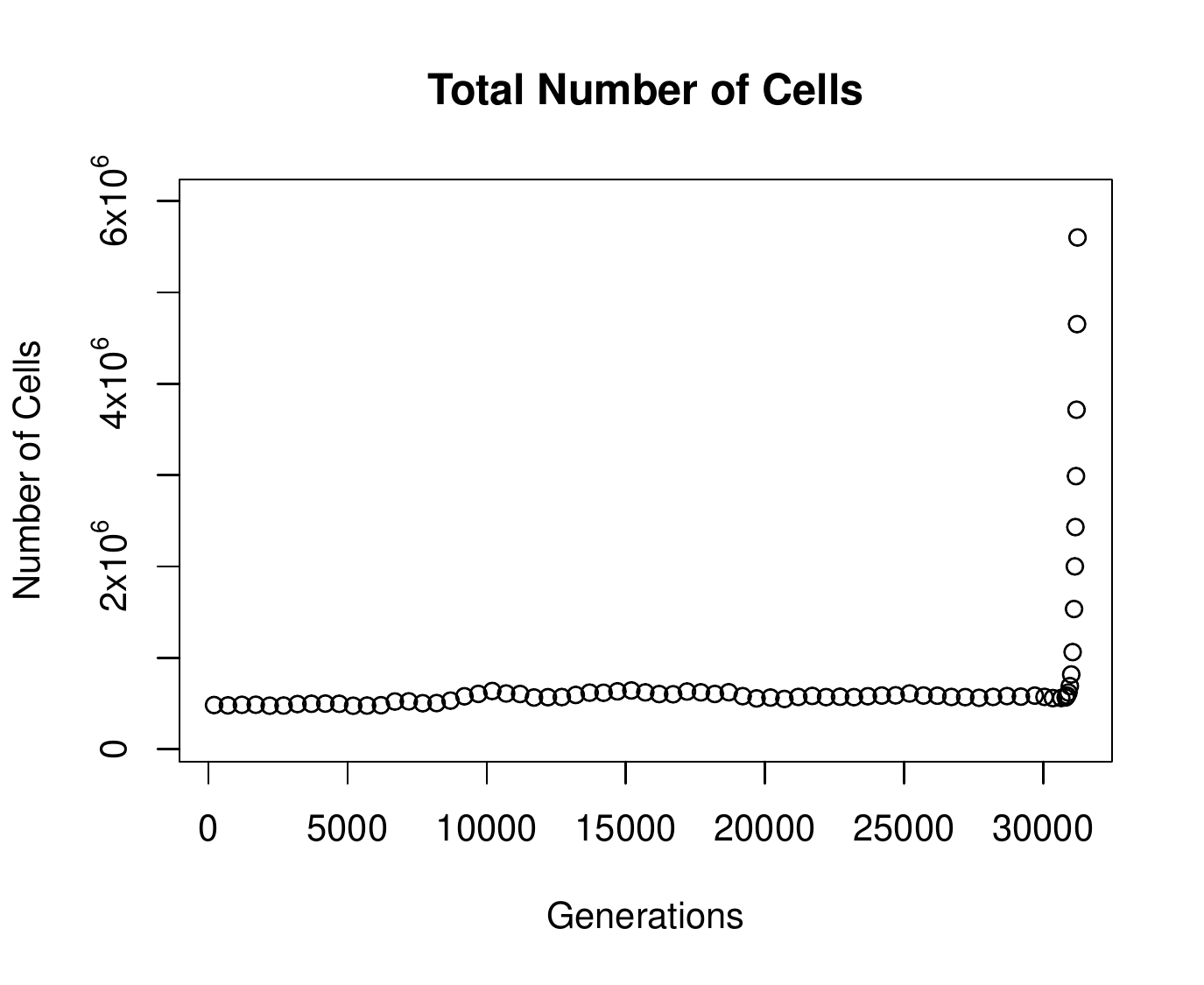}
 \caption{The dynamics of the total number of cells.
 Initially, the total cell count increases only marginally but at some point, a combination of primary and secondary driver mutations within one cell with a large fitness benefit arises and leads to rapid exponential proliferation
 (parameters: 
Initial number of cells $N=500000$, 
secondary driver fitness advantage $s_\mathrm{P}=10^{-5}$, 
the primary driver fitness advantage $s_{\mathrm{D}^+}=0.05$, 
primary driver disadvantage $s_{\mathrm{D}^-}=0.1$, 
advantage of a secondary driver mutation in the presence of the primary driver mutation $s_\mathrm{DP}=0.015$,
mutation rates for secondary driver mutations $\mu_\mathrm{P}=2\cdot10^{-5}$,
mutation rate for the primary driver mutation $\mu_\mathrm{D}=5\cdot10^{-6}$).
 }
 \label{fig:total}
\end{center}
\end{figure}
\begin{figure}[htb]
\begin{center}
\includegraphics[width=1.0\textwidth]{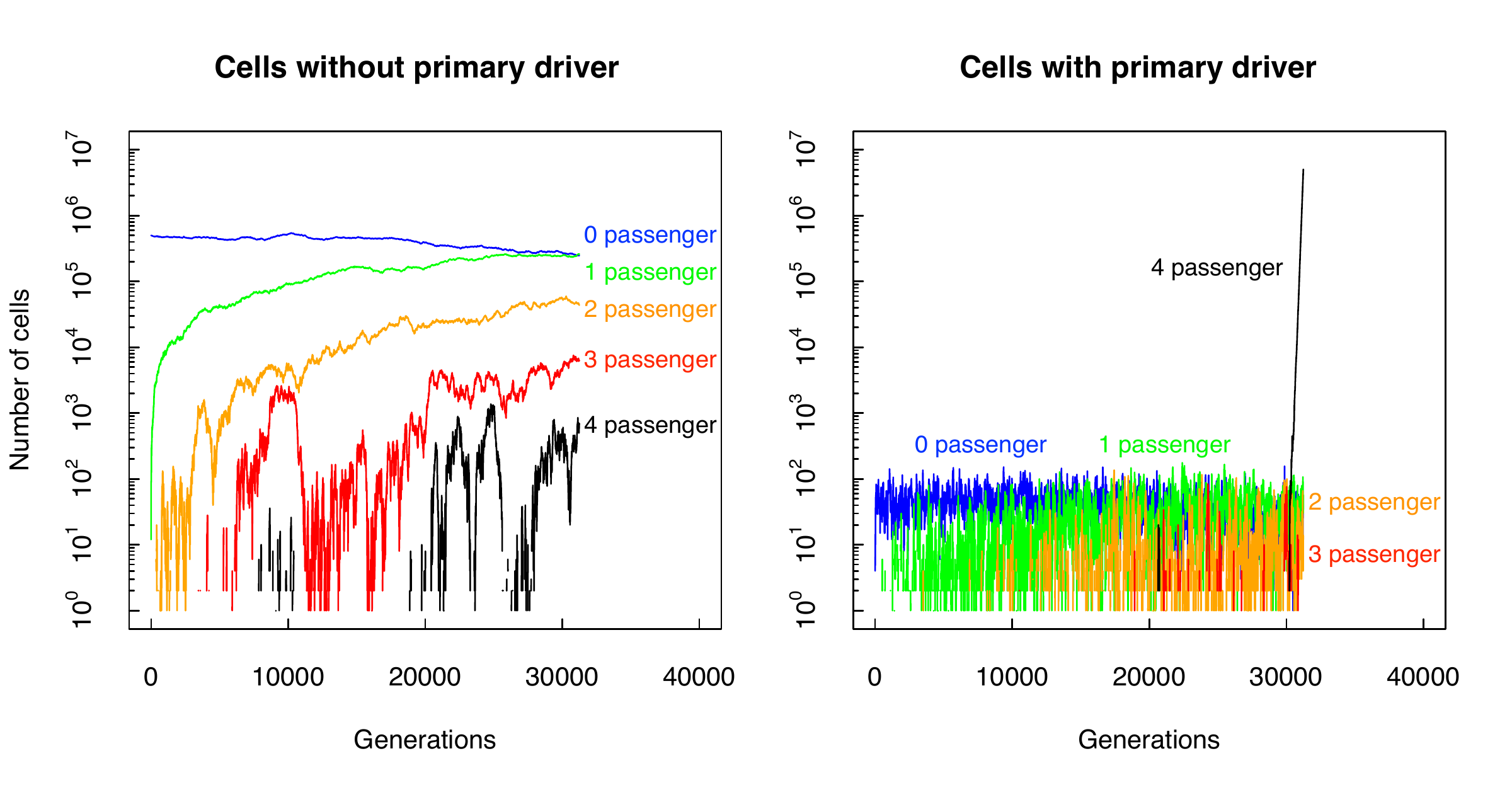}
\caption{
The dynamics of the number of cells with different numbers of mutations in a single simulation.
Top: The number of cells without any mutation decreases slightly, whereas the number of cells with secondary driver mutations, but no primary driver mutation, slowly increases. 
Bottom: While a small number cells with the primary driver mutation is present from the beginning, at first these primary driver mutations are not accompanied by sufficiently many secondary driver mutations to compensate the disadvantage arising from the primary driver. Only when a primary driver mutation is co-occurring with enough secondary driver mutations (in this case four), the number of cells with the primary driver starts to increase rapidly. 
(parameters: 
$N=500000$, 
$s_\mathrm{P}=10^{-5}$, 
$s_{\mathrm{D}^+}=0.05$, 
$s_{\mathrm{D}^-}=0.1$, 
$s_\mathrm{DP}=0.015$,
$\mu_\mathrm{P}=2\cdot10^{-5}$,
$\mu_\mathrm{D}=5\cdot10^{-6}$)}
\label{fig:single}
\end{center}
\end{figure}
In Figure \ref{fig:single}, the total number of cells is subdivided into the number of cells with different numbers of mutations. 
The left panel presents the cells that have not acquired the primary driver mutation, 
the right one shows cells with the primary driver mutation.
Cells with the primary driver mutation, but not enough secondary driver mutations, arise occasionally, but those cells die out quickly again -- thus, their average abundance is small. 
Cells without the primary driver mutation do not die out, they also do not induce fast growth, cf.\ Figure \ref{fig:single}.
Only cells that have obtained enough secondary driver mutations and in addition acquire the primary driver mutation, divide so quickly that the population size increases rapidly.

The parameters in our figures have been chosen such that a cells acquires a substantial growth advantage once the primary driver mutation co-occurs with 4 secondary driver mutations.
This event can occur at any time and hence, in some simulation the number of cells can increase very early, whereas in other simulations the number of cells does not undergo fast proliferation for many generations.
Consequently, the rate of progression has an enormous variation. 
For the parameters from our figures, the time at which rapid proliferation occurs varied between $\approx 9300$ and $\approx 63000$ generations in 500 simulations. The distribution of these times is discussed in more detail below.

\subsection{Analytical results}
\subsubsection{Average Number of Cells}
We can calculate the average number of cells with a certain number of mutations at a given generation $t$. The number of cells which do not have the primary driver mutation and $k$ secondary driver mutations (i.e. $x_{0,k}(t)$) changes on average by means of the cell's fitness and it decreases by the mutation rate
\begin{align}
x_{0,k}(t) & = (1-(\mu_\mathrm{P}+\mu_\mathrm{D}))(1+s_\mathrm{P})^{k} x_{0,k}(t-1) + \mu_\mathrm{P}(1+s_\mathrm{P})^{k-1}x_{0,k-1}(t-1),
\label{eq:x0k1pre4}
\end{align}
where $x_{0,-1}(t) \equiv 0$.
The solution of Equation \eqref{eq:x0k1pre4} 
for $s_\mathrm{P} \neq 0$, i.e.\
if the secondary driver mutations are not neutral, is
\begin{align}\label{eq:pass2}
x_{0,k}(t) = N\mu_\mathrm{P}^k (1-(\mu_\mathrm{D}+\mu_\mathrm{P}))^{t-k}(1+s_\mathrm{P})^{k(k-1)/2} \prod_{i=0}^{k-1}\frac{1-(1+s_\mathrm{P})^{t-i}}{1-(1+s_\mathrm{P})^{i+1}},
\end{align}
where $N$ denotes the initial number of cells. 
The mathematical proof of Equation \eqref{eq:pass2} is given in \ref{subsec:withoutDriver}.
Note, that the product can be written in terms of a $q$-binomial coefficient \citep{koekoek:book:2010},
\begin{align}
\setstretch{1}
\prod_{i=0}^{k-1}\frac{1-(1+s_\mathrm{P})^{t-i}}{1-(1+s_\mathrm{P})^{i+1}} = \begin{bmatrix} t\\k \end{bmatrix}_{1+s_\mathrm{P}}.
\end{align}
For the case $s_\mathrm{P} = 0$, we take the limit of the $q$-binomial coefficient \citep[e.g.][]{kac:book:2002}
\begin{align}
\setstretch{1}
 \lim_{s_\mathrm{P}\rightarrow 0} \begin{bmatrix} t\\k \end{bmatrix}_{1+s_\mathrm{P}} = \binom{t}{k}
\end{align}
and obtain
\begin{align}\label{eq:pass1}
 x_{0,k}(t) = N\mu_\mathrm{P}^k (1-(\mu_\mathrm{D}+\mu_\mathrm{P}))^{t-k} \binom{t}{k},
 \end{align}
which is the result that is also expected if the secondary driver mutations are neutral and accumulated independently of each other.

Intuitively, the term $\mu_\mathrm{P}^k(1-(\mu_\mathrm{D}+\mu_\mathrm{P}))^{t-k}$ describes the probability of obtaining exactly $k$ mutations in $t$ generations. 
There are different possibilities when the mutations happen, these possibilities are captured
by the binomial coefficient $\binom{t}{k}$. 
Thus, we have a growing polynomial term in $t$ and a declining exponential term in $t$, since $(1-(\mu_\mathrm{D}+\mu_\mathrm{P}))<1$.

In the case of $s_\mathrm{P}\neq0$, the interpretation is similar. 
Here, additionally the fitness advantage for secondary driver mutations has to be taken into account. Since the number of cells with $j$ secondary driver mutations grows with $(1+s_\mathrm{P})^j$, also the number of cells that can mutate grows. Hence, the factor $(1+s_\mathrm{P})^{k(k-1)/2}$ is multiplied to the expression and the binomial coefficient turns into the $q$-binomial coefficient.

\begin{table}[h]
\begin{center}
\label{tbl:abbrev}
\begin{tabular}{ll}
$\mu_\mathrm{P}$ & Mutation rate for secondary driver mutations \\
$\mu_\mathrm{D}$ & Mutation rate for the primary driver mutation \\
 $\nu_\mathrm{P} = 1-\mu_\mathrm{D}-\mu_\mathrm{P}$ & Probability for a cell without the primary driver mutation to not mutate \\
$\nu_{\mathrm{D}} = 1-\mu_\mathrm{P}$ & Probability for a cell with the primary driver mutation to not mutate \\
$s_\mathrm{P}$ & Fitness change of a secondary driver mutation (see \eqref{eq:ratesb0}, \eqref{eq:ratesb1}) \\
$s_{\mathrm{D}^+}$ & Fitness advantage of the primary driver mutation (see \eqref{eq:ratesb0}, \eqref{eq:ratesb1})\\
$s_{\mathrm{D}^-}$ & Fitness disadvantage of the primary driver mutation (see \eqref{eq:ratesb0}, \eqref{eq:ratesb1})\\
$s_{\mathrm{DP}}$ & Fitness advantage of combination of a secondary driver and the primary\\
&  driver mutation \\
$ \varsigma_\mathrm{P} = 1+s_\mathrm{P}$ & Fitness according to a secondary driver without the primary driver mutation \\
$\varsigma_\mathrm{DP}=1+s_\mathrm{DP}$ & Fitness according to the combination of primary and secondary driver mutation \\
$\varsigma_\mathrm{D}=\frac{1+s_\mathrm{D}^+}{1+s_\mathrm{D}^-} \approx 1+s_\mathrm{D}^+-s_\mathrm{D}^-$ & Fitness according to a primary driver mutation with no secondary driver
\end{tabular}
\caption{Summary of our abbreviations}
\end{center}
\end{table}

For cells that have obtained the primary driver mutation and $k$ secondary driver mutations, the situation is slightly more complex.
There are $k+1$ different possibilities on how to obtain $k$ secondary driver and the primary driver mutation, since some of the secondary driver mutations may have occurred before the primary driver mutation has been acquired
whereas others may have occurred afterwards. 
Let $x_{1,k}^{(p)}(t)$ denote the number of cells with the primary driver mutation and $k$ secondary driver mutations, when the primary driver mutation has happened in a cell with $p$ secondary driver mutations. 
Note that $0 \leq p \leq k$. 
The change in the number of cells now depends on $p$. 
Using the abbreviations from Table \ref{tbl:abbrev} to simplify our notation, we have
\begin{align}
 x_{1,k}^{(p)}(t) = \begin{cases} \nu_\mathrm{D}\varsigma_\mathrm{D}\varsigma_\mathrm{DP}^k x_{1,k}^{(p)}(t-1) + \mu_\mathrm{P}\varsigma_\mathrm{D}\varsigma_\mathrm{DP}^{k-1} x_{1,k-1}^{(p)}(t-1), \quad & \mbox{if } p<k \\
                                                         \nu_\mathrm{D}\varsigma_\mathrm{D}\varsigma_\mathrm{DP}^k x_{1,k}^{(p)}(t-1) + \mu_\mathrm{D} \varsigma_\mathrm{P}^{k} x_{0,k}(t-1), & \mbox{if } p=k. \end{cases}
 \label{eq:x1kppre}
\end{align}
To express the average number of cells in total we need to sum over all possible pathways,
\begin{align}\label{eq:driversum}
 x_{1,k}(t) = \sum_{p=0}^k x_{1,k}^{(p)}(t).
\end{align}
In \ref{subsec:withDriver}, we proof that the analytical solution of Equation \eqref{eq:driversum} is
\begin{align}\label{eq:driv2}
\setstretch{1}
\begin{split}
 x_{1,k}(t)=&N\sum_{p=0}^{k}\mu_D\mu_P^k \varsigma_\mathrm{D}^{k-p} \varsigma_\mathrm{DP}^{(k(k-1)-p(p-1))/2}\frac{\varsigma_\mathrm{P}^{p(p+1)/2}}{\prod_{n=0}^{p-1}(1-\varsigma_\mathrm{P}^{n+1})}
\\
  \times & \left[ \nu_\mathrm{P}^{t-p}\left( \sum_{r=0}^{p} \frac{\left(-\varsigma_\mathrm{P}^{t-p+1}\right)^r \varsigma_\mathrm{P}^{\frac{r(r-1)}{2}}}{\prod_{n=p}^{k}(\nu_\mathrm{P}\varsigma_\mathrm{P}^r-\nu_\mathrm{D}\varsigma_\mathrm{D}\varsigma_\mathrm{DP}^n)} \begin{bmatrix}p\\r\end{bmatrix}_{\varsigma_\mathrm{P}}  \right)  \right.   \\
  - & \sum_{j=p}^{k}\nu_\mathrm{P}^{j-p}(\nu_\mathrm{D}\varsigma_\mathrm{D}\varsigma_\mathrm{DP}^j)^{t-k} \prod_{m=j}^{k-1} \frac{1-\varsigma_\mathrm{DP}^{t-m-1}}{1-\varsigma_\mathrm{DP}^{k-m}}\\
  \times & \left. \left( \sum_{r=0}^{p} \frac{\left(-\varsigma_\mathrm{P}^{j-p+1}\right)^r \varsigma_\mathrm{P}^{\frac{r(r-1)}{2}}}{\prod_{n=p}^{j}(\nu_\mathrm{P}\varsigma_\mathrm{P}^r-\nu_\mathrm{D}\varsigma_\mathrm{D}\varsigma_\mathrm{DP}^n)} \begin{bmatrix}p\\r\end{bmatrix}_{\varsigma_\mathrm{P}} \right) \right],
  \end{split}
\end{align}
if $s_\mathrm{P} \neq 0$.
The summation over $p$ indicates the different mutational pathways.
An intuitive explanation of this somewhat lengthy equation is given in \ref{sec:Intuitive}.

Interestingly, the case for $s_\mathrm{P}=0$ is much more challenging. 
The underlying problem is that the normal binomial coefficient cannot be expressed in a sum in the way the $q$-binomial coefficient can be expressed \citep{koekoek:book:2010},
\begin{align}\label{eq:qbinomsum}
\setstretch{1}
 \begin{bmatrix}t\\p\end{bmatrix}_{\varsigma_\mathrm{P}} 
 = \prod_{j=0}^{p-1} \frac{1-\varsigma_\mathrm{P}^{t-j}}{1-\varsigma_\mathrm{P}^{j+1}}
 = \frac{\sum_{r=0}^{p} \left( -\varsigma_\mathrm{P}^t\right)^r\left( 1/\varsigma_\mathrm{P}\right)^{\frac{r(r-1)}{2}} \begin{bmatrix} p\\r \end{bmatrix}_{1/\varsigma_\mathrm{P}}}{\prod_{j=0}^{p-1}(1-\varsigma_\mathrm{P}^{j+1})}.
\end{align}
When summing over all generations of the population with $p$ secondary driver mutations to derive the expression for the population of cells with $p+1$ secondary driver mutations, we have to calculate the sum
\begin{align}
 \sum_{i=0}^{t-p-1} \left(\varsigma_\mathrm{D}\varsigma_\mathrm{DP}^p\right)^i\binom{t-i-1}{p}.
\end{align}
When we go further and try to calculate the expression for the population with $k$ secondary driver mutations, we need to apply this sum $(k-p)$-times and hence we obtain a multi sum,
\begin{align}
\label{eq:multisum}
 \sum_{i_0=0}^{t-k-1} \left(\varsigma_\mathrm{D}\varsigma_\mathrm{DP}^{k}\right)^{i_0} \sum_{i_1=0}^{t-k-i_0-2}\left(\varsigma_\mathrm{D}\varsigma_\mathrm{DP}^{k-1}\right)^{i_1} \cdots& \sum_{i_{k-p}=0}^{t-2k+p-i_0-i_1-\cdots -i_{k-p-1}-1} \left(\varsigma_\mathrm{D}\varsigma_\mathrm{DP}^{p}\right)^{i_{k-p}} \\
 \nonumber
&\times \binom{t-2k+p -i_0-i_1- \cdots -i_{k-p-1} -1}{p}.
\end{align}
Only an analytical expression for this multi sum would allow a closed solution of the problem with $s_\mathrm{P}=0$.
Also taking the limit $s_\mathrm{P} \to 0$ of our expression for  $s_\mathrm{P}\neq0$ is a substantial mathematical challenge. 
However, we can use our solution for $s_\mathrm{P}\neq0$ for arbitrarily small values of $s_\mathrm{P}$. 
Moreover, numerical considerations show that the result for $s_\mathrm{P}=0$ is very close to the case of $s_\mathrm{P}\ll 1$.

In Figure \ref{fig:dyn}, the dynamics of the average number of cells with a certain number of mutations, is shown, both without and with the primary driver mutation.
Simulation results for $s_\mathrm{P}=0$ agree very well with the analytical result obtained for $s_\mathrm{P}\neq0$. 

\begin{figure}[htp]
\begin{center}
 \includegraphics[width=1.0\textwidth]{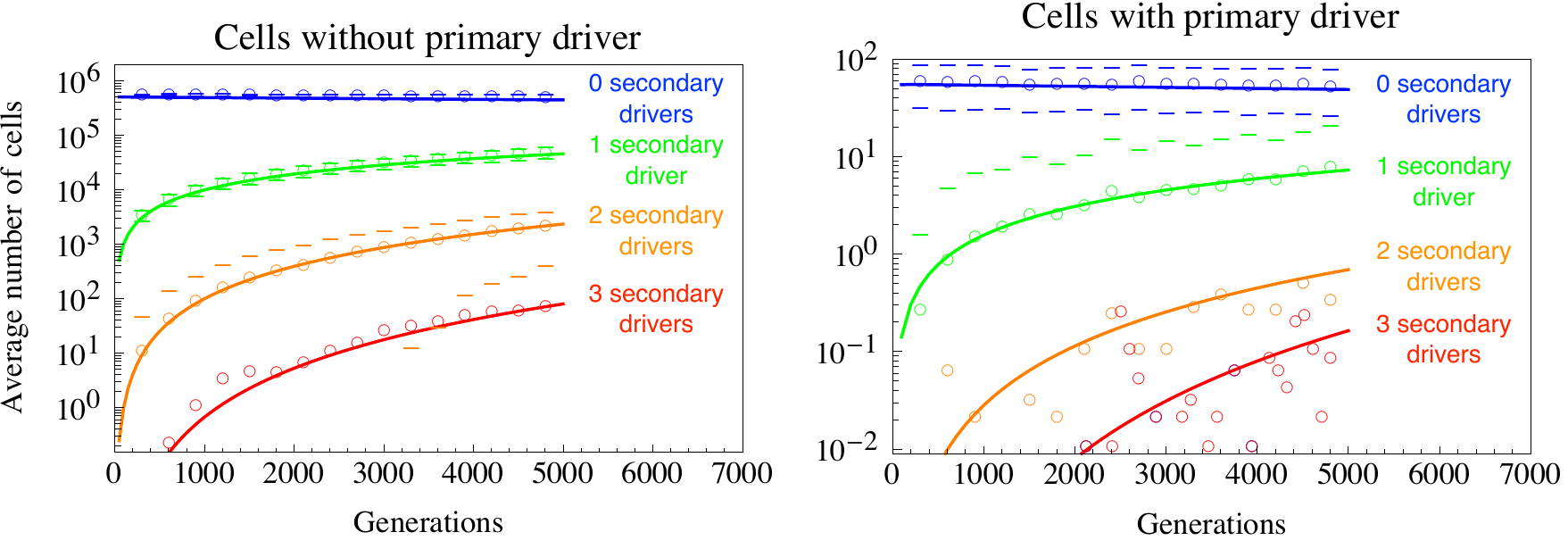}
 \caption{Dynamics of the number of cells with different number of secondary driver mutations, without (left) and with (right) the primary driver mutation. 
Simulation results averaged over 500 independent realizations for $s_\mathrm{P}=0$ (circles) agree almost perfectly with the analytical result obtained for $s_\mathrm{P} = 10^{-5}$.
The bars represent the standard deviation. 
Cells with no mutation have a very small relative standard deviation and cells with one mutation (i.e.\ one passenger only or the driver only) have a relatively small standard deviation. 
In contrast, cells with two passenger mutations for instance have a very broad standard deviation in the beginning that is approximately four times the average number. 
Only in few realizations, a primary driver mutations co-occurs with several secondary drivers, hence the simulation data for these cases shows a large spread 
(parameters: 
$N=500000$, 
$s_\mathrm{P}=10^{-5}$, 
$s_{\mathrm{D}^+}=0.05$, 
$s_{\mathrm{D}^-}=0.1$, 
$s_\mathrm{DP}=0.015$,
$\mu_\mathrm{P}=2\cdot10^{-5}$,
$\mu_\mathrm{D}=5\cdot10^{-6}$).}
 \label{fig:dyn}
 \end{center}
\end{figure}

\subsubsection{Distribution of time until cancer initiation}
Next, let us calculate the distribution of the time it takes until rapid proliferation occurs. 
Since we use a multi type, time discrete branching process, we can make use of the probability generating function to recursively calculate the probability for a certain cell type cell to be present at a certain time $t$ \citep{haccou:book:2005}.
Of particular interest is the probability that a tumor initiating cell is present, i.e.\ in the example above a cell with the primary driver mutation and 4 secondary driver mutations. 
Let $f_{i,j}(s_{0,0}, \ldots, s_{0,k},s_{1,0},\ldots,s_{1,k})$ be the probability generation function for the cell type $x_{i,j}$ in the branching process described above. 
The probability that there is no cell with the primary driver mutation and $k$ secondary driver mutations at time $t$, when the process starts with one $x_{0,0}$-cell, can be computed by recursively calculating $f(t) :=f_{0,0}^{\circ (t)} (1,\ldots, 1,0)$, where
\begin{align}
f_{i,j}^{\circ(t)}(\textbf{s}) = f_{i,j}^{}(f_{0,0}^{\circ(t-1)}(\textbf{s}),\ldots,f_{1,k}^{\circ(t-1)}(\textbf{s})).
\end{align}
Starting with $N$ cells, we have to consider $f(t)^N$. 
Thus, the probability that there is at least one tumor cell at time $t$ is $1-f(t)^N$. 
If we are interested in the probability $\tau(t)$ that there is at least one tumor cell exactly at time $t$, we need to subtract the probability, that there has been a tumor cell before $\tau(t)=1-f(t)^N - 1 + f(t-1)^N = f(t-1)^N-f(t)^N$. Figure \ref{fig:timecompare} shows the good agreement between simulations and the recursive calculation.
\begin{figure}
 \centering \includegraphics[width=0.7\textwidth]{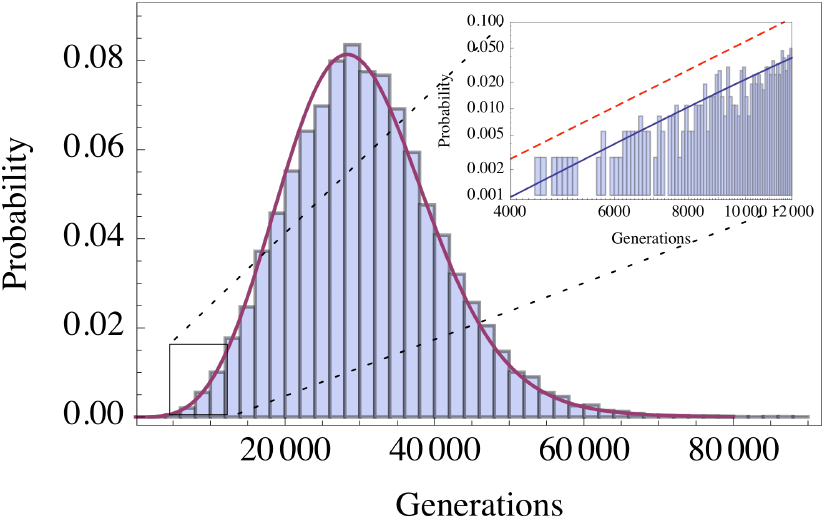}
 \caption{Comparison between the analytical calculation and simulations of the distribution until cancer initiation (main panel). 
 Solid lines represents analytic solution. 
 The analytical calculation and simulations of the model agree very well.
 This inset illustrates that the time distribution initially follows a power law with an exponent of $\approx 3.4$ shown as a dashed line
 (parameters: 
$N=500000$, 
$s_\mathrm{P}=10^{-5}$, 
$s_{\mathrm{D}^+}=0.05$, 
$s_{\mathrm{D}^-}=0.1$, 
$s_\mathrm{DP}=0.015$,
$\mu_\mathrm{P}=2\cdot10^{-5}$,
$\mu_\mathrm{D}=5\cdot10^{-6}$,
 distribution over 20000 independent realizations). 
 }
 \label{fig:timecompare}
\end{figure}

The time distribution for low $t$ follows a power law, as shown in the inset of Figure \ref{fig:timecompare}. 
The exponent of the power law is approximately $3.4$. 
If all mutations were neutral, one would expect a lead coefficient of approximately 4 to accumulate five mutations, 
as derived by \citet{armitage:BJC:1954}. 
In our case, the curve increases slower. Numerical considerations show that 
the main reason for this is that, in contrast to \citep{armitage:BJC:1954}, we allow extinction: 
Many lineages that have accumulated mutations go extinct before the final, 
cancer causing mutation arises.

\section*{Discussion}

Most models in literature assume that each 
mutation leads to an independent and steady increase in the cells' fitness \citep{beerenwinkel:PlosCB:2007,bozic:PNAS:2010,gerstung:MPS:2010,reiter:EA:2013,michor:NRC:2004}. 
In this context, neutral passenger mutations have no causal impact on cancer progression.
Only recently, some authors have considered passenger mutations not only as neutral byproducts of the clonal expansion of mutagenic cells, but as having a deleterious impact on the cells' fitness \citep{mcfarland:PNAS:2013}.

Here, we have described a model in which the fitness of the driver mutation strongly depends on the number of passenger mutations the cell has acquired. 
These passenger mutations, which we have termed secondary driver mutations, 
lead only to a small change in fitness or no change in fitness at all.
As illustrated in Figures \ref{fig:total}, \ref{fig:single}, and \ref{fig:dyn}, 
the number of cells stays roughly constant for a long time before it rapidly increases, despite the fact that mutations occur in the process permanently. 
This dynamic effect of cancer initiation is very different from models in which mutations do not interact with each other. 
We speculate that this kind of dynamics can have important implications for diagnosis and treatment.
In principle, the dynamics presented in Figure \ref{fig:total} can also be the result of one highly advantageous, but very unlikely driver mutation. 
But in such a case, cells with the driver mutation should not be present in the population before tumorigenesis. 
This contradicts with current knowledge about the {\it MYC} translocation which has also been detected  in humans without lymphoma \citep{muller:PNAS:1995}. 
This effect is well captured by our model, as shown in Figure \ref{fig:single}.

In some tumors, such as Burkitt Lymphoma, the neoplasms is only diagnosed after fast tumor growth has started.
In this case, sequencing studies have shown that several mutations are present at the time of examination 
\citep{schmitz:Nature:2012,alexandrov:Nature:2013,richter:NG:2012,love:NG:2012}.
Since the patients typically do not have any symptoms in before diagnosis of the cancer, 
it is possible that some mutations have virtually no direct impact on the cells fitness. 
Nevertheless, they are necessary for the initiation of the cancer, as they indirectly allow the driver mutation to initiate rapid cell growth. This agrees well with our epistatic model, where (nearly) neutral secondary driver mutations occur at a fixed rate before the cancer can be diagnosed. 

Of course, not all mutations have such an epistatic effect on primary driver mutations, some might even be considered deleterious \citep{mcfarland:PNAS:2013}. 
Nevertheless, our work shows that mutations that appear to be neutral in one context should not only be regarded as a neutral byproduct of the clonal expansion of mutagenic cells.
Instead, in some cases passenger mutations can have a serious impact in cancer initiation, in particular when there are non-trivial interactions between different mutations.
In this case the term ``passenger" may not be the most appropriate one.
To understand the impact of those interactions can be essential for a deeper understanding of the initiation of cancer.
\section*{Acknowledgments}
We thank J.\ Richter for stimulating discussions on Burkitt Lymphoma and B.\ Werner for helpful comments on our manuscript. 
Generous funding by the Max-Planck Society is gratefully acknowledged.
R.S.\ is supported the German Ministry of Education and Science (BMBF) through the MMML-MYC-SYS network Systems Biology of MYC-positive lymphomas (036166B) 

\appendix

\section{ \bf Analytic expression for the average number of cells without the primary driver mutation at generation $t$}\label{subsec:withoutDriver}

\subsection{Secondary driver fitness advantage is unequal to zero - $k$ secondary driver mutations}
We first consider the case without the primary driver mutation. We assume $s_\mathrm{P}\neq 0$ (and consequently $\varsigma_\mathrm{P} \neq 1$), as discussed in the main text. The rate change of cells with $k$ secondary driver mutations is
\begin{align}
x_{0,k}(t) & = \nu_{\mathrm{P}}\varsigma_\mathrm{P}^{k} x_{0,k}(t-1) + \mu_\mathrm{P}\varsigma_\mathrm{P}^{k-1}x_{0,k-1}(t-1),
\label{eq:x0k1pre}
\end{align}
where $x_{0,-1}(t) \equiv 0$.
The solution of \eqref{eq:x0k1pre} is formulated in the following theorem:

\begin{thm}\label{thm:x0k1}
For any integer $k \geq 0$, the number of cells with $k$ secondary driver mutations and no primary driver mutation is
\begin{align}
 x_{0,k}(t) = N\mu_\mathrm{P}^k \nu_{\mathrm{P}}^{t-k}\varsigma_\mathrm{P}^{k(k-1)/2} \prod_{n=0}^{k-1}\frac{1-\varsigma_\mathrm{P}^{t-n}}{1-\varsigma_\mathrm{P}^{n+1}}.
 \label{eq:x0k1}
\end{align}
\end{thm}

\begin{proof}
If \eqref{eq:x0k1} is a solution, then it must satisfy \eqref{eq:x0k1pre}. Since solutions for recursive functions are always unique, \eqref{eq:x0k1} would be the only solution. Hence, we proof \eqref{eq:x0k1} by inserting the equation on the right hand side of \eqref{eq:x0k1pre}.

\begin{align}
\begin{split}
 & \nu_{\mathrm{P}}\varsigma_\mathrm{P}^{k} x_{0,k}(t-1) + \mu_\mathrm{P}\varsigma_\mathrm{P}^{k-1}x_{0,k-1}(t-1)\\
 &\quad = \nu_{\mathrm{P}}\varsigma_\mathrm{P}^{k} N\mu_\mathrm{P}^k \nu_{\mathrm{P}}^{t-k-1}\varsigma_\mathrm{P}^{k(k-1)/2} \prod_{n=0}^{k-1}\frac{1-\varsigma_\mathrm{P}^{t-n-1}}{1-\varsigma_\mathrm{P}^{n+1}} + \mu_\mathrm{P}\varsigma_\mathrm{P}^{k-1}N\mu_\mathrm{P}^{k-1} \nu_{\mathrm{P}}^{t-k}\varsigma_\mathrm{P}^{(k-1)(k-2)/2} \prod_{n=0}^{k-2}\frac{1-\varsigma_\mathrm{P}^{t-n-1}}{1-\varsigma_\mathrm{P}^{n+1}} \\
 & \quad = N\mu_\mathrm{P}^k \nu_{\mathrm{P}}^{t-k}\varsigma_\mathrm{P}^{k(k-1)/2} \left( \varsigma_\mathrm{P}^{k} \prod_{n=0}^{k-1}\frac{1-\varsigma_\mathrm{P}^{t-n-1}}{1-\varsigma_\mathrm{P}^{n+1}} + \prod_{n=0}^{k-2}\frac{1-\varsigma_\mathrm{P}^{t-n-1}}{1-\varsigma_\mathrm{P}^{n+1}} \right).
\end{split}
\label{eq:x0kpr}
\end{align}
We can write each product as a $q$-binomial coefficient, $\prod_{n=0}^{k-1}\frac{1-\varsigma_\mathrm{P}^{t-n}}{1-\varsigma_\mathrm{P}^{n+1}} = \begin{bmatrix} t\\k \end{bmatrix}_{\varsigma_\mathrm{P}}$. Thus, with the $q$-Pascal rule \citep{kac:book:2002}
\begin{align}
 \varsigma_\mathrm{P}^{k}\begin{bmatrix} t-1\\k \end{bmatrix}_{\varsigma_\mathrm{P}} +\begin{bmatrix} t-1\\k-1 \end{bmatrix}_{\varsigma_\mathrm{P}} = \begin{bmatrix} t\\k \end{bmatrix}_{\varsigma_\mathrm{P}}
\end{align}
Equation \eqref{eq:x0kpr} simplifies to
\begin{align}
\nu_{\mathrm{P}}\varsigma_\mathrm{P}^{k} x_{0,k}(t-1) + \mu_\mathrm{P}\varsigma_\mathrm{P}^{k-1}x_{0,k-1}(t-1) = N\mu_\mathrm{P}^k \nu_{\mathrm{P}}^{t-k}\varsigma_\mathrm{P}^{k(k-1)/2} \prod_{n=0}^{k-1}\frac{1-\varsigma_\mathrm{P}^{t-n}}{1-\varsigma_\mathrm{P}^{n+1}} = x_{0,k}(t),
\end{align}
which concludes the proof.
\end{proof}

\section{\bf Analytic expression for the average number of cells with the primary driver mutation at generation $t$}\label{subsec:withDriver}

We now turn to the cells which have obtained the primary driver mutation. As discussed in the main text, we only look at the case where the fitness change of the secondary driver mutation is not equal to zero, $s_\mathrm{P}\neq 0$. While for cells without the primary driver mutation there is only one mutational pathway, cells with the primary driver mutation can be reached via different mutational pathways, because cells that get the primary driver mutation might have different amounts of secondary driver mutations. Hence, we need to sum over all those possible pathways. Let $p$ be the number of secondary drivers that are present in the cell which acquires the primary driver mutation. Then $x_{1,k}^{(p)}(t)$ denotes the number of cells with the primary driver mutation and $k$ secondary driver mutations, when the primary driver mutation has happened in a cell with $p$ secondary driver mutations ($0 \leq p \leq k$). With this, the total number of cells with the primary driver mutation is 
\begin{align}
 x_{1,k}(t) = \sum_{p=0}^k x_{1,k}^{(p)}(t).
\end{align}

The change in the number of cells now depends on $p$. We have
\begin{align}
 x_{1,k}^{(p)}(t) = \begin{cases} \nu_\mathrm{D}\varsigma_\mathrm{D}\varsigma_\mathrm{DP}^k x_{1,k}^{(p)}(t-1) + \mu_\mathrm{P} \varsigma_\mathrm{D}\varsigma_\mathrm{DP}^{k-1} x_{1,k-1}^{(p)}(t-1), \quad & \mbox{if } p<k \\
                                                         \nu_\mathrm{D}\varsigma_\mathrm{D}\varsigma_\mathrm{DP}^k x_{1,k}^{(p)}(t-1) + \mu_\mathrm{D} \varsigma_\mathrm{P}^{k} x_{0,k}(t-1), & \mbox{if } p=k. \end{cases}
 \label{eq:x1kppre}
\end{align}

The solution of \eqref{eq:x1kppre} is given by the following theorem:
\begin{thm} \label{thm:x1kp}
The average number of cells with the primary driver mutation and $k$ secondary driver mutations, given that the primary driver mutation happens in a cell with $p$ secondary driver mutations, is  given by
\begin{align}
\begin{split}
  x_{1,k}^{(p)}(t) =& N\mu_D\mu_P^k \varsigma_\mathrm{D}^{k-p} \varsigma_\mathrm{DP}^{(k(k-1)-p(p-1))/2} \frac{\varsigma_\mathrm{P}^{p(p+1)/2}}{\prod_{n=0}^{p-1}\left(1-\varsigma_\mathrm{P}^{n+1}\right)} \\
  \times & \left[ \nu_{\mathrm{P}}^{t-p} \Psi_{p,k}(t) - \sum_{j=p}^{k}\nu_{\mathrm{P}}^{j-p}\left(\nu_\mathrm{D}\varsigma_\mathrm{DP}^j\varsigma_\mathrm{D}\right)^{t-k} \Psi_{p,j}(j)\prod_{n=j}^{k-1}\frac{1-\varsigma_\mathrm{DP}^{t-n-1}}{1-\varsigma_\mathrm{DP}^{k-n}} \right],
\end{split}
  \label{driverexp}
\end{align}
where the function $\Psi$ is defined
\begin{align}\label{eq:psi}
 \Psi_{p,k}(t) = \sum_{r=0}^{p} \frac{\left( -\varsigma_\mathrm{P}^{t-p+1}\right)^r \varsigma_\mathrm{P}^{\frac{r(r-1)}{2}} \begin{bmatrix} p\\r \end{bmatrix}_{\varsigma_\mathrm{P}}}{\prod_{j=p}^{k}(\nu_\mathrm{P}\varsigma_\mathrm{P}^r-\nu_\mathrm{D}\varsigma_\mathrm{D}\varsigma_\mathrm{DP}^j)}.
\end{align}
\end{thm}

\begin{proof} Again we proof the theorem by inserting \eqref{driverexp} in \eqref{eq:x1kppre} and showing that the equality holds true. We 
need to distinguish between the two cases as in \eqref{eq:x1kppre}. First, we proof the theorem for the case $p<k$.
\begin{align} \label{eq:x1kppr1}
\begin{split}
 & \nu_\mathrm{D}\varsigma_\mathrm{D}\varsigma_\mathrm{DP}^k x_{1,k}^{(p)}(t-1) + \mu_\mathrm{P}\varsigma_\mathrm{D}\varsigma_\mathrm{DP}^{k-1} x_{1,k-1}^{(p)}(t-1)\\ 
 & \quad = \nu_\mathrm{D}\varsigma_\mathrm{D}\varsigma_\mathrm{DP}^k N\mu_D\mu_P^k \varsigma_\mathrm{D}^{k-p} \varsigma_\mathrm{DP}^{(k(k-1)-p(p-1))/2} \frac{\varsigma_\mathrm{P}^{p(p+1)/2}}{\prod_{n=0}^{p-1}\left(1-\varsigma_\mathrm{P}^{n+1}\right)} \\
 & \quad \times \left[ \nu_{\mathrm{P}}^{t-p-1} \Psi_{p,k}(t-1) - \sum_{j=p}^{k}\nu_{\mathrm{P}}^{j-p}\left(\nu_\mathrm{D}\varsigma_\mathrm{DP}^j\varsigma_\mathrm{D}\right)^{t-k-1} \Psi_{p,j}(j)\prod_{n=j}^{k-1}\frac{1-\varsigma_\mathrm{DP}^{t-n-2}}{1-\varsigma_\mathrm{DP}^{k-n}} \right] \\
 & \quad + \mu_\mathrm{P}\varsigma_\mathrm{D}\varsigma_\mathrm{DP}^{k-1} N\mu_D\mu_P^{k-1} \varsigma_\mathrm{D}^{k-p-1} \varsigma_\mathrm{DP}^{((k-1)(k-2)-p(p-1))/2} \frac{\varsigma_\mathrm{P}^{p(p+1)/2}}{\prod_{n=0}^{p-1}\left(1-\varsigma_\mathrm{P}^{n+1}\right)} \\
 & \quad \times \left[ \nu_{\mathrm{P}}^{t-p-1} \Psi_{p,k-1}(t-1) - \sum_{j=p}^{k-1}\nu_{\mathrm{P}}^{j-p}\left(\nu_\mathrm{D}\varsigma_\mathrm{DP}^j\varsigma_\mathrm{D}\right)^{t-k} \Psi_{p,j}(j)\prod_{n=j}^{k-2}\frac{1-\varsigma_\mathrm{DP}^{t-n-2}}{1-\varsigma_\mathrm{DP}^{k-n-1}} \right] \\
 & \quad = N\mu_D\mu_P^k \varsigma_\mathrm{D}^{k-p} \varsigma_\mathrm{DP}^{(k(k-1)-p(p-1))/2} \frac{\varsigma_\mathrm{P}^{p(p+1)/2}}{\prod_{n=0}^{p-1}\left(1-\varsigma_\mathrm{P}^{n+1}\right)} \\ 
 & \quad \times \left[ \nu_\mathrm{P}^{t-p-1}\left( \nu_\mathrm{D}\varsigma_\mathrm{D}\varsigma_\mathrm{DP}^k \Psi_{p,k}(t-1) + \Psi_{p,k-1}(t-1) \right)  - \nu_\mathrm{P}^{k-p}(\nu_\mathrm{D}\varsigma_\mathrm{D}\varsigma_\mathrm{DP}^k)^{t-k}\Psi_{p,k}(k) \right.\\
 & \quad - \left. \sum_{j=p}^{k-1} \nu_\mathrm{P}^{j-p} \Psi_{p,j}(j) \left(\nu_\mathrm{D}\varsigma_\mathrm{D}\varsigma_\mathrm{DP}^j \right)^{t-k}\left( \varsigma_\mathrm{DP}^{k-j}\prod_{n=j}^{k-1}\frac{1-\varsigma_\mathrm{DP}^{t-n-2}}{1-\varsigma_\mathrm{DP}^{k-n}}+\prod_{n=j}^{k-2}\frac{1-\varsigma_\mathrm{DP}^{t-n-2}}{1-\varsigma_\mathrm{DP}^{k-n-1}} \right) \right]
 \end{split}
\end{align}
When we compare \eqref{driverexp} and \eqref{eq:x1kppr1} we see, that the two equations are equal if 
\begin{align}
 \nu_\mathrm{D}\varsigma_\mathrm{D}\varsigma_\mathrm{DP}^k \Psi_{p,k}(t-1) + \Psi_{p,k-1}(t-1) = \nu_{\mathrm{P}} \Psi_{p,k}(t)
 \label{eq:psipr}
\end{align}
and
\begin{align}\label{eq:s}
 \varsigma_\mathrm{DP}^{k-j}\prod_{n=j}^{k-1}\frac{1-\varsigma_\mathrm{DP}^{t-n-2}}{1-\varsigma_\mathrm{DP}^{k-n}}+\prod_{n=j}^{k-2}\frac{1-\varsigma_\mathrm{DP}^{t-n-2}}{1-\varsigma_\mathrm{DP}^{k-n-1}} = \prod_{n=j}^{k-1}\frac{1-\varsigma_\mathrm{DP}^{t-n-1}}{1-\varsigma_\mathrm{DP}^{k-n}}
\end{align}
For Equation \eqref{eq:s}, we have
\begin{align}
\nonumber
 & & \varsigma_\mathrm{DP}^{k-j}\prod_{n=j}^{k-1}\frac{1-\varsigma_\mathrm{DP}^{t-n-2}}{1-\varsigma_\mathrm{DP}^{k-n}}+\prod_{n=j}^{k-2}\frac{1-\varsigma_\mathrm{DP}^{t-n-2}}{1-\varsigma_\mathrm{DP}^{k-n-1}} & = \prod_{n=j}^{k-1}\frac{1-\varsigma_\mathrm{DP}^{t-n-1}}{1-\varsigma_\mathrm{DP}^{k-n}} \\
 &\Leftrightarrow & \varsigma_\mathrm{DP}^{k-j} \prod_{n=j}^{k-1} \left( 1-\varsigma_\mathrm{DP}^{t-n-2} \right) + \left(1-\varsigma_\mathrm{DP}^{k-j}\right) \prod_{n=j}^{k-2}\left( 1-\varsigma_\mathrm{DP}^{t-n-2} \right) & = \prod_{n=j}^{k-1} \left( 1-\varsigma_\mathrm{DP}^{t-n-1} \right) \\
 \nonumber
 &\Leftrightarrow & \left( 1-\varsigma_\mathrm{DP}^{k-j} + \varsigma_\mathrm{DP}^{k-j}\left( 1-\varsigma_\mathrm{DP}^{t-k-1} \right) \right) \prod_{n=j}^{k-2}\left( 1-\varsigma_\mathrm{DP}^{t-n-2} \right) & = \prod_{n=j}^{k-1} \left( 1-\varsigma_\mathrm{DP}^{t-n-1} \right) \\
 \nonumber
 &\Leftrightarrow & \prod_{n=j}^{k-1} \left( 1-\varsigma_\mathrm{DP}^{t-n-1} \right) & = \prod_{n=j}^{k-1} \left( 1-\varsigma_\mathrm{DP}^{t-n-1} \right).
\end{align}
For Equation \eqref{eq:psipr}, we need to insert the definition of $\Psi$
\begin{align}
 & \nu_\mathrm{D}\varsigma_\mathrm{D}\varsigma_\mathrm{DP}^k \Psi_{p,k}(t-1) + \Psi_{p,k-1}(t-1) \\
 & \quad = \sum_{r=0}^{p} \left(-\varsigma_\mathrm{P}^{t-p}\right)^r \varsigma_\mathrm{P}^{r(r-1)/2} \begin{bmatrix} p\\r\end{bmatrix}_{\varsigma_\mathrm{P}} \left( \frac{\nu_\mathrm{D}\varsigma_\mathrm{D}\varsigma_\mathrm{DP}^k}{\prod_{j=p}^{k} \left( \nu_\mathrm{P}\varsigma_\mathrm{P}^r - \nu_\mathrm{D}\varsigma_\mathrm{D}\varsigma_\mathrm{DP}^j \right)} + \frac{1}{\prod_{j=p}^{k-1} \left( \nu_\mathrm{P}\varsigma_\mathrm{P}^r - \nu_\mathrm{D}\varsigma_\mathrm{D}\varsigma_\mathrm{DP}^j \right)} \right) \\
 & \quad = \sum_{r=0}^{p} \left(-\varsigma_\mathrm{P}^{t-p}\right)^r \varsigma_\mathrm{P}^{r(r-1)/2} \begin{bmatrix} p\\r\end{bmatrix}_{\varsigma_\mathrm{P}} \left( \frac{ \nu_\mathrm{D}\varsigma_\mathrm{D}\varsigma_\mathrm{DP}^k + \nu_\mathrm{P}\varsigma_\mathrm{P}^r - \nu_\mathrm{D}\varsigma_\mathrm{D}\varsigma_\mathrm{DP}^k}{\prod_{j=p}^{k} \left( \nu_\mathrm{P}\varsigma_\mathrm{P}^r - \nu_\mathrm{D}\varsigma_\mathrm{D}\varsigma_\mathrm{DP}^j \right)} \right) \\
 & \quad = \nu_\mathrm{P} \sum_{r=0}^{p} \left(-\varsigma_\mathrm{P}^{t-p+1}\right)^r \varsigma_\mathrm{P}^{r(r-1)/2} \begin{bmatrix} p\\r\end{bmatrix}_{\varsigma_\mathrm{P}} \\
 & \quad = \nu_\mathrm{P} \Psi_{p,k}(t).
\end{align}
This concludes the proof for the case $p<k$. Now we look at the case $p=k$. We have
\begin{align}
 & \nu_\mathrm{D}\varsigma_\mathrm{D}\varsigma_\mathrm{DP}^k x_{1,k}^{(k)}(t-1) + \mu_\mathrm{D} \varsigma_\mathrm{P}^{k} x_{0,k}(t-1)\\
 & \quad = \nu_\mathrm{D}\varsigma_\mathrm{D}\varsigma_\mathrm{DP}^k N \mu_\mathrm{D}\mu_\mathrm{P}^k \frac{\varsigma_\mathrm{P}^{k(k+1)/2}}{\prod_{n=0}^{k-1}(1-\varsigma_\mathrm{P}^{n+1})} \left[ \nu_\mathrm{P}^{t-k-1} \Psi_{k,k}(t-1) - \left(\nu_\mathrm{D}\varsigma_\mathrm{D}\varsigma_\mathrm{DP}^k\right)^{t-k-1} \Psi_{k,k}(k) \right]\\
 & \quad + \mu_\mathrm{D}\varsigma_\mathrm{P}^k N\mu_\mathrm{P}^k \nu_\mathrm{P}^{t-k-1} \varsigma_\mathrm{P}^{k(k-1)/2} \prod_{n=0}^{k-1}\frac{1-\varsigma_\mathrm{P}^{t-n-1}}{1-\varsigma_\mathrm{P}^{n+1}}\\
 & \quad = N \mu_\mathrm{D}\mu_\mathrm{P}^k \frac{\varsigma_\mathrm{P}^{k(k+1)/2}}{\prod_{n=0}^{k-1}(1-\varsigma_\mathrm{P}^{n+1})}\left[ \nu_\mathrm{D}\varsigma_\mathrm{D}\varsigma_\mathrm{DP}^k\nu_\mathrm{P}^{t-k-1} \Psi_{k,k}(t-1) - \left(\nu_\mathrm{D}\varsigma_\mathrm{D}\varsigma_\mathrm{DP}^k\right)^{t-k} \Psi_{k,k}(k) + \nu_\mathrm{P}^{t-k-1}\prod_{n=0}^{k-1}\left( 1-\varsigma_\mathrm{P}^{t-n-1} \right) \right].
\end{align}
In order for this to be equal to $x_{1,k}^{(k)}$, we need
\begin{align}
 \nu_\mathrm{D}\varsigma_\mathrm{D}\varsigma_\mathrm{DP}^k \Psi_{k,k}(t-1)+ \prod_{n=0}^{k-1}\left( 1-\varsigma_\mathrm{P}^{t-n-1} \right) = \nu_\mathrm{P} \Psi_{k,k}(t).
\end{align}
Analogue to \eqref{eq:psipr} this equation holds true if 
\begin{align}
  \prod_{n=0}^{k-1}\left( 1-\varsigma_\mathrm{P}^{t-n-1} \right)=\Psi_{k-1,k}(t-1)=\sum_{r=0}^{k} \left(-\varsigma_\mathrm{P}^{t-k}\right)^r \varsigma_\mathrm{P}^{r(r-1)/2} \begin{bmatrix} p\\r \end{bmatrix}_{\varsigma_\mathrm{P}}.
\end{align}
By writing the summation as a $q$-Pochhammer symbol, we have
\begin{align}
 \sum_{r=0}^{k} \left(-\varsigma_\mathrm{P}^{t-k}\right)^r \varsigma_\mathrm{P}^{r(r-1)/2} \begin{bmatrix} p\\r \end{bmatrix}_{\varsigma_\mathrm{P}} = \left(\varsigma_\mathrm{P}^{t-k}; \varsigma_\mathrm{P}\right)_k = \prod_{n=0}^{k-1}\left(1-\varsigma_\mathrm{P}^{t-k+n} \right) = \prod_{n=0}^{k-1}\left( 1-\varsigma_\mathrm{P}^{t-n-1} \right).
\end{align}
This concludes the proof also for $p=k$.

\end{proof}

\section{\bf Intuitive Description of Equation (11)}\label{sec:Intuitive}
Here, we try to understand this equation in a more intuitive way. 
For each generation $t$, the number of possibilities to distribute the $p$ secondary driver mutations over $t$ time steps is given by the $q$-binomial coefficient $\begin{bmatrix}t\\p\end{bmatrix}_{\varsigma_\mathrm{P}}$. But the growth of the cells depends on the time when the secondary driver mutations are first acquired. 
Due to fitness advantage, the earlier the mutations have been acquired, the faster the population grows, and also the sooner the primary driver mutation can be obtained. 
As in Equation (5), the effect of the fitness advantage on the cells without the primary driver mutation itself is captured by multiplying $\varsigma_\mathrm{P}^{p(p+1)/2}$. 
The effect on the primary driver mutation is more intricate.
To capture this effect, we start from a $q$-binomial coefficient and rewrite the $q$-Pochhammer symbol in the numerator
$\prod_{j=0}^{p-1} 1-\varsigma_\mathrm{P}^{t-j}$ in terms of a sum \citep{koekoek:book:2010},
\begin{align}\label{eq:qbinomsum}
 \begin{bmatrix}t\\p\end{bmatrix}_{\varsigma_\mathrm{P}} 
 = \prod_{j=0}^{p-1} \frac{1-\varsigma_\mathrm{P}^{t-j}}{1-\varsigma_\mathrm{P}^{j+1}}
 = \frac{\sum_{r=0}^{p} \left( -\varsigma_\mathrm{P}^t\right)^r\left( 1/\varsigma_\mathrm{P}\right)^{\frac{r(r-1)}{2}} \begin{bmatrix} p\\r \end{bmatrix}_{1/\varsigma_\mathrm{P}}}{\prod_{j=0}^{p-1}(1-\varsigma_\mathrm{P}^{j+1})}.
\end{align}
To make this resemble the term in the parentheses in the second line of Equation (11), we divide the numerator by
$\prod_{j=p}^{k}(\nu_\mathrm{P}\varsigma_\mathrm{P}^r-\nu_\mathrm{D}\varsigma_\mathrm{D}\varsigma_\mathrm{DP}^j)$
and we obtain
\begin{align}\label{eq:a}
  \frac{\sum_{r=0}^{p} \frac{1}{\prod_{j=p}^{k}(\nu_\mathrm{P}\varsigma_\mathrm{P}^r-\nu_\mathrm{D}\varsigma_\mathrm{D}\varsigma_\mathrm{DP}^j)}\left( -\varsigma_\mathrm{P}^t\right)^r\left( 1/\varsigma_\mathrm{P}\right)^{\frac{r(r-1)}{2}} \begin{bmatrix} p\\r \end{bmatrix}_{1/\varsigma_\mathrm{P}}}{\prod_{j=0}^{p-1}(1-\varsigma_\mathrm{P}^{j+1})}.
\end{align}
With 
\begin{align}
 \begin{bmatrix} p\\r \end{bmatrix}_{\varsigma_\mathrm{P}}= \prod_{j=0}^{r-1}\frac{1-\varsigma_\mathrm{P}^{p-j}}{1-\varsigma_\mathrm{P}^{j+1}}=\prod_{j=0}^{r-1}\frac{\varsigma_\mathrm{P}^{p-2j-1}(1-1/\varsigma_\mathrm{P}^{p-j})}{1-1/\varsigma_\mathrm{P}^{j+1}}=\varsigma_\mathrm{P}^{r(p-r)}\prod_{j=0}^{r-1}\frac{1-1/\varsigma_\mathrm{P}^{p-j}}{1-1/\varsigma_\mathrm{P}^{j+1}}=\varsigma_\mathrm{P}^{r(p-r)}\begin{bmatrix} p\\r \end{bmatrix}_{1/\varsigma_\mathrm{P}}
\end{align}
Equation \eqref{eq:a} can be written as
\begin{align}
 \frac{\sum_{r=0}^{p} \frac{1}{\prod_{j=p}^{k}(\nu_\mathrm{P}\varsigma_\mathrm{P}^r-\nu_\mathrm{D}\varsigma_\mathrm{D}\varsigma_\mathrm{DP}^j)}\left( -\varsigma_\mathrm{P}^{t-p+1}\right)^r \varsigma_\mathrm{P}^{\frac{r(r-1)}{2}} \begin{bmatrix} p\\r \end{bmatrix}_{\varsigma_\mathrm{P}}}{\prod_{j=0}^{p-1}(1-\varsigma_\mathrm{P}^{j+1})}.
\end{align}
For the numerator of this modified $q$-binomial coefficient, we introduce the abbreviation 
\begin{align}\label{eq:psi}
 \Psi_{p,k}(t) = \sum_{r=0}^{p} \frac{\left( -\varsigma_\mathrm{P}^{t-p+1}\right)^r \varsigma_\mathrm{P}^{\frac{r(r-1)}{2}} \begin{bmatrix} p\\r \end{bmatrix}_{\varsigma_\mathrm{P}}}{\prod_{j=p}^{k}(\nu_\mathrm{P}\varsigma_\mathrm{P}^r-\nu_\mathrm{D}\varsigma_\mathrm{D}\varsigma_\mathrm{DP}^j)}.
\end{align}
In terms of this $\Psi$-function Equation (11) can be written in a more compact form as
\begin{align}
\begin{split}
 x_{1,k}^{}(t) =& N\sum_{p=0	}^{k}\mu_D\mu_P^k \varsigma_\mathrm{D}^{k-p} \varsigma_\mathrm{DP}^{(k(k-1)-p(p-1))/2} \frac{\varsigma_\mathrm{P}^{p(p+1)/2}}{\prod_{j=0}^{p-1}(1-\varsigma_\mathrm{P}^{j+1})} 
 \\
  \times & \left[ \nu_\mathrm{P}^{t-p} \Psi_{p,k}(t) - \sum_{j=p}^{k}\nu_\mathrm{P}^{j-p}(\nu_\mathrm{D} \varsigma_\mathrm{D} \varsigma_\mathrm{DP}^j
  )^{t-k} \Psi_{p,j}(j)\prod_{m=j}^{k-1}\frac{1-\varsigma_\mathrm{DP}^{t-m-1}}{1-\varsigma_\mathrm{DP}^{k-m}} \right].
 \end{split}
\end{align}

\end{document}